\documentclass[pra,aps,twocolumn,superscriptaddress,10pt,showpacs]{revtex4-1}

\usepackage{latexsym}
\usepackage{amsmath}
\usepackage{amssymb}
\usepackage{amsthm}

\newtheorem{propo}{Proposition}
\newtheorem{lem}{Lemma}

\newcommand{\ket}[1]{ | #1 \rangle}

\renewcommand{\epsilon}{\varepsilon}

\begin{document}

\title{Detecting genuine multipartite entanglement of pure states with bipartite correlations}

\author{Marcin Markiewicz} 
\affiliation{Institute of Theoretical Physics and Astrophysics, University of Gda\'nsk, 80-952 Gda\'nsk, Poland}
\author{Wies\l aw~Laskowski} 
\affiliation{Institute of Theoretical Physics and Astrophysics, University of Gda\'nsk, 80-952 Gda\'nsk, Poland}
\author{Tomasz Paterek}
\affiliation{School of Physical and Mathematical Sciences, Nanyang Technological University, Singapore}
\affiliation{Centre for Quantum Technologies, National University of Singapore, Singapore}
\author{Marek \.Zukowski}
\affiliation{Institute of Theoretical Physics and Astrophysics, University of Gda\'nsk, 80-952 Gda\'nsk, Poland}

\begin{abstract}
Monogamy of bipartite correlations leads, for arbitrary  pure multi-qubit states,  to simple conditions able to indicate various types of  multipartite entanglement by being capable to  exclude the possibility of $k$-separability.
\end{abstract}


\maketitle


For bipartite systems the phenomenon of quantum entanglement \cite{HHH, PAN} manifests itself in  correlations.
One might expect that genuinely $n$-partite entanglement gives rise to non-vanishing correlations between all $n$ subsystems.
This  is incorrect, at least when correlations are quantified as average values of a product of local measurement results \cite{DAG,BENNETT2012}
(for a discussion on quantum correlations without classical correlations, see {\em e.g.} \cite{RMP}).
Thus, in order to detect genuine multipartite entanglement of certain states one has to rely on correlations between smaller number of subsystems \cite{WIESNIAK}.

Here we discuss global features of multiparty qubit pure states, which can be deduced from  their bipartite correlations. We shall use  the property of monogamy of correlations \cite{KPRLK10}.
Another approach has recently been put forward by W\"urflinger \emph{et al.} \cite{WBAGV12} who have shown that some non-entangled reduced density operators can be linked with global entangled states.

Monogamy of quantum correlations can be used for entanglement detection \cite{SCHMIDT}.
We will  use  it to derive a criterion for \emph{genuine} multipartite entanglement.
Monogamy was also employed in studies  of quantum marginal problem, {\em i.e.} conditions for existence of a global quantum state given its marginals \cite{K04}, security of quantum key distribution \cite{VALERIO},  and to show  that correlations between macroscopic measurements ought to be classical~\cite{MACRO}. It leads to efficient methods of solving strongly correlated multipartite quantum lattice systems \cite{VMC08}.

States of $n$ qubits (two-level quantum systems) have density matrices of the following form:
\begin{equation}
\rho=\frac{1}{2^n}\sum_{\mu_1,...,\mu_n=0,1,2,3}T_{\mu_1... \mu_n}\sigma_{\mu_1}\otimes...\otimes \sigma_{\mu_n},
\label{cortensor}
\end{equation}
where $\sigma_{\mu_k} \in \{\openone, \sigma_1, \sigma_2, \sigma_3 \}$. $\sigma_1, \sigma_2$ and $\sigma_3$ are the Pauli matrices of the $k$th observer related to the $\hat x$, $\hat y$ and $\hat z$ local directions of her local Cartesian basis (we allow each observer to choose her  {\em own} coordinates, 
this will play a {\em crucial} role).
The values $T_{\mu_1...\mu_n}$ are given by a correlation function for measurements of the Pauli operators:
$T_{\mu_1...\mu_n}=\textrm{Tr} \left( \rho\, \sigma_{\mu_1}\otimes...\otimes \sigma_{\mu_n} \right)$.
We call the whole object endowed with components $T_{\mu_1... \mu_n} $  an extended correlation tensor  $\hat T$. 
Under rotations of local Cartesian coordinates, its components with $d$ zeros  transform like $(n-d)$-order tensors. The values of $T_{i0...0}$, {\em i.e.} for all indices $0$ except $i=1,2,3$ for the first qubit indices, give us the components of the Bloch vector $\vec{b}^{1}$ defining the reduced desity matrix of the qubit. Similar identifications with Bloch vectors, denoted as $\vec{b}^{(k)}$, hold for all other qubits.

A pure $n$-partite state $|\psi\rangle$ is called a $k$-product one, if it can be represented as a tensor product of $k$ pure $r_m$-partite states:
\begin{equation}
|\psi_{k\mathrm{-prod}}\rangle=|\psi_{r_1}\rangle\otimes\ldots\otimes|\psi_{r_k}\rangle.
\label{kprodukt}
\end{equation}
Of course $\sum_{m=1}^k r_m=n$.
There are different types of $k$-product states corresponding to different ways of partitioning $n$ into a sum of $k$ integers.
We will refer to a  definite type of $k$-product state as $(r_1+\ldots+r_k)$-partition product state.
For example, a $n=4$ partite state can be a $2$-product in two ways, $(3+1)$ and $(2+2)$, and it can be a $3$-product in one way, $(2+1+1)$.
It is clear that if a state is not $k$-product it also cannot be $k'$-product for $k' \ge k$.
If a state is not $k$-product, it can be at most $(k-1)$-product.
The number of mutually entangled particles is minimized when entanglement is distributed as uniformly as possible  (the maximal number of particles share an entangled state).
Thus,  a state that is not $k$-product has a subset of at least $\left\lceil n/(k-1)\right\rceil$ mutually entangled particles.
For example, if a $7$-partite state is not $3$-product, it can be at most biseparable. The number of mutually entangled particles is  minimal for partition $(3+4)$. Thus, we have entanglement between at least $4$ particles.


The quantity which plays the main role in our method is the sum of squares of all possible bipartite correlations.  Monogamy relations lead to the following property.
\begin{propo}
For any $n$-qubit state (pure or mixed) the following tight bound holds:
\begin{equation}
\label{M ident}
\mathcal M=\sum_{1\leq k<l\leq n}\mathcal M_{kl}\leq \begin{cases} 2 &\mbox{if } n=2 \\
{n \choose 2} & \mbox{if } n\geq 3 \end{cases},
\end{equation}
with
\begin{equation}
\label{Mkl}
\mathcal M_{kl}=\sum_{i,j=1,2}T^2_{0,\ldots,0,i_{(k)},0,\ldots,0,j_{(l)},0,\ldots,0},  
\end{equation}
where subscripts $(k)$ and $(l)$ denote $k$-th and $l$-th position, and  $i,j$ are two pairs of  Cartesian coordinate indices. For simplicity, we shall assume that they always represent coordinate indexes related with directions $x$ and $y$. 
\label{teor M}
\end{propo}
Note that ${n \choose 2}$ is  the number of terms $\mathcal M_{kl}$ in $\mathcal M$.

\emph{Proof}.
Our thesis for $n=2$ and $n=3$ follows  from the following monogamy relations, which are direct generalizations of the ones derived in \cite{KPRLK10}:
\begin{eqnarray}
\mathcal M_{kl} &\leq& 2, \quad \textrm{for all } k \neq l, \label{MONO1} \\
\mathcal M_{kl}+\mathcal M_{lm}&\leq& 2, \quad \textrm{for all } k \neq l \neq m, \label{MONO2} \\
\mathcal M_{kl}+\mathcal M_{lm}+\mathcal M_{km}&\leq& 3, \quad \textrm{for all } k \neq l \neq m. \label{MONO3}
\end{eqnarray}

The relations (\ref{MONO1}) and (\ref{MONO2}) can be derived using the following property of a set of operators shown in Refs.~\cite{KPRLK10,TG,WW1}:
Let $S=\{\hat A_1,\ldots,\hat A_j\}$ be a set of Hermitian, traceless operators fullfilling:
$\hat A_k\hat A_l+\hat A_l\hat A_k=0$ and $\hat A_k^2=\openone$, 
and let $\vec \alpha=(\langle\hat A_1\rangle,\ldots,\langle\hat A_j\rangle)$ be a vector of their expectation values. Then the following holds:
\begin{equation}
||\vec \alpha||^2\leq 1.
\label{pr7}
\end{equation}

Without endagering generality we can make a proof of (\ref{MONO1}) for just two qubits, say $1$-st and $2$-nd. Define two sets of operators:
$
S_1=\{\sigma_1\otimes\sigma_1,
\sigma_1\otimes\sigma_2\}$
 and
$S_2=\{\sigma_2\otimes\sigma_1,
\sigma_2\otimes\sigma_2\}$.
It is straightforward to see that the operators within each set fulfill all the above assumptions. However, 
$
||\vec \alpha_1||^2=T_{11}^2+T_{12}^2$ and 
$||\vec \alpha_2||^2=T_{21}^2+T_{22}^2$. Therefore, 
$
||\vec \alpha_1||^2+||\vec \alpha_2||^2=\mathcal M_{12}\leq 2.
$

For the proof of (\ref{MONO2}) let us limit ourselves to qubits $1,2$ and $3$, and define
\begin{eqnarray}
S'_1=\{\sigma_1\otimes\sigma_1\otimes\openone ,\sigma_1\otimes\sigma_2\otimes\openone, \nonumber \\
\sigma_2\otimes\openone\otimes\sigma_1 , \sigma_2\otimes\openone\otimes\sigma_2   \},\\
S'_2=\{\sigma_2\otimes\sigma_1\otimes\openone ,\sigma_2\otimes\sigma_2\otimes\openone, \nonumber \\
\sigma_1\otimes\openone\otimes\sigma_1 , \sigma_1\otimes\openone\otimes\sigma_2  \}\\  \nonumber
\end{eqnarray}
Both sets fulfill assumptions leading to ineq. (\ref{pr7}), thus after a similar identification of averages of the operators with components of the generalized correlation tensor, we get
$
\mathcal M_{12}+\mathcal M_{13}\leq 2.
$

The inequality (\ref{MONO3}) follows from summing up the following three inequalities, each of the form of (\ref{MONO2}):
\begin{eqnarray}
\mathcal M_{kl}+\mathcal M_{lm}&\leq& 2, \nonumber \\
\mathcal M_{kl}+\mathcal M_{km}&\leq& 2, \nonumber \\
\mathcal M_{lm}+\mathcal M_{km}&\leq& 2.
\end{eqnarray}

For higher $n$ we combine inequalities (\ref{MONO2}) and (\ref{MONO3}) to arrive at the thesis. Let us put  $\mathcal M$ as follows:
\begin{eqnarray}
&&\mathcal M=(\mathcal M_{12}+\mathcal M_{13}+...+\mathcal M_{1n})+(\mathcal M_{23}+...+\mathcal M_{2n})\nonumber\\
&&+\ldots+(\mathcal M_{n-2,n-1}+\mathcal M_{n-2,n})+(\mathcal M_{n-1,n}),
\label{Mx}
\end{eqnarray}
where in every bracket we group the terms that share a common subsystem corresponding to the first index.
Consider first the case of even total number of $\mathcal M_{kl}$ terms in $\mathcal M$, {\em i.e.} an even ${n \choose 2}$.
In the brackets which contain an even number of $\mathcal M_{kl}$ terms, we split the sum into consecutive pairs and apply inequality (\ref{MONO2}) to every pair.
In the brackets with odd number of $\mathcal M_{kl}$ terms, we split the sum into consecutive pairs and the last term.
As the total number of $\mathcal M_{kl}$ terms is even,  the number of ``last terms'' is also even.
The last terms have a common last subsystem.  We again group them in pairs. We have $\frac{1}{2}{n \choose 2}$ pairs.
By the inequality (\ref{MONO2}) each  pair is upper bounded by $2$. Thus,  the sum is bounded by  ${n \choose 2}$ .

For an  odd  number of $\mathcal M_{kl}$ terms in $\mathcal M$, i.e. ${n \choose 2}$ odd, we first apply inequality (\ref{MONO3}) to its last three terms (in Eq. (\ref{Mx})  these terms are shown explicitly).
The number of remaining $\mathcal M_{kl}$ terms is even, and we proceed as before:
We use $\frac{1}{2}[{n \choose 2} - 3]$ inequalities (\ref{MONO2}) and one inequality (\ref{MONO3}). This again gives the  bound  ${n \choose 2}$ .

The bound  is tight. Take a  state $\otimes_{k=1}^n\ket{+}_{k} $, where $\ket{+}_k$ is the $+1$ eigenvalue eigenstate of  $\sigma_x$ of the $k$-th observer. For the state all ${n \choose 2}$ correlation tensor elements with two $1$ indices entering each $\mathcal M_{kl}$  are equal to $1$. $\Box$

Since the bound of $\mathcal M$ is attained by a product state,  $\mathcal M$ seems to be useless as an entanglement identifier.
However, this can be overcome by a choice of a suitable local coordinate system for each of the observers.
Namely, we shall say that the $k$-th observer uses her {\em preferred} Cartesian basis, if the Bloch vector of her qubit is pointing in the $z$ direction. 
For such set of local coordinates  the sum in $\mathcal M_{kl}$, see eq. (\ref{Mkl}), will be always zero for the above discussed product state, as the local Bloch vectors of the reduced density matrices, which were earlier $\hat{x}$, in  the prefered coordinates are by definition $\hat{z}$, and due to the nature of the state, the values entering $\mathcal M_{kl}$ factorize. For example, $T_{ij0...0}=T_{i00...0}T_{0j0...0}=b^{(1)}_ib^{(2)}_j$. Thus components for $i=1,2$ and $j=1,2$ in (\ref{Mkl}) are zero.

 Let us denote $\mathcal M$ in the new preferred set of coordinate systems by $\mathcal M^{(pb)}$. It is a sum of bipartite correlations involving only tensor components related to local  directions {\em orthogonal} to the local Bloch vectors. From now on the local observers do not have any freedom to define the $\hat{z}$ direction - it is set by the Bloch vector of the local reduced density operator. 
If a given qubit has vanishing Bloch vector, any axis can serve as the local $\hat z$ direction.

The quantity $\mathcal M^{(pb)}$ has the following property, which makes it useful as entanglement identifier.
\begin{propo}
For a given type of $k$-product state $|\psi_{k\mathrm{-prod}}\rangle = \otimes_{m=1,..,k}|\psi^{r_m}\rangle$  the following property holds:
\begin{equation}
\mathcal M^{(pb)}(\otimes_{m=1,..,k}|\psi^{r_m}\rangle) = \sum_{m=1}^{k}\mathcal M^{(pb)}(|\psi^{r_m}\rangle).
\label{N k-sep}
\end{equation}
\label{prop sum}
\end{propo}
\begin{proof}
Note that in the sum $\mathcal M^{(pb)}(|\psi^{r_1}\rangle...|\psi^{r_k}\rangle)$ correlation tensor elements in formula (\ref{Mkl}) with index $i$ belonging to one subsystem (say $r_a$) and index $j$ belonging to another (say $r_b$) can be factorized, since they are effectively calculated for the product state $|\psi^{r_a}\rangle \otimes |\psi^{r_b}\rangle$. Thus,
\begin{equation}
T_{0,\ldots,0,i,0,\ldots,0,j,0,\ldots,0}=T^{(r_a)}_{0,\ldots,0,i,0,\ldots,0}T^{(r_b)}_{0,\ldots,0,j,0,\ldots,0}.
\end{equation}
But the factors on the right-hand-side are $x$ or $y$ components of  the single particle Bloch vectors, and thus vanish in the preferred Cartesian bases.  Thus, correlations between subsystems $r_a$ and $r_b$ are not present in $\mathcal M^{(pb)}(|\psi^{r_1}\rangle...|\psi^{r_k}\rangle)$. The sum can be decomposed into sums solely within the subsystems form $r_1$ to $r_k$. That is, we have (\ref{N k-sep}).\end{proof}

The following Proposition is a basis of our method of identifying entanglement:
\begin{propo}
For $n\geq 3$, and for any pure $n$-qubit state $|\psi\rangle$,   and any class $\mathcal S$ of k-product states  of the type $r_1+...+r_k=n$ the following holds:
\begin{equation}
|\psi\rangle \in \mathcal S \Longrightarrow \mathcal M^{(pb)}(|\psi\rangle)\leq\sum_{m=1}^k {r_m \choose 2}+d,
\label{PROP_KSEP}
\end{equation}
where $d$ is the number of terms in the sequence $r_1,\ldots,r_k$ with $r_m=2$. 
\label{teor k-sep}
\end{propo}

\emph{Proof}.
For any state $|\psi\rangle \in \mathcal S$ we have
\begin{eqnarray}
&\mathcal M^{(pb)}(|\psi\rangle)\leq\max_{|\psi_{\mathcal S}\rangle \in  \mathcal S} \mathcal M^{(pb)}(|\psi_{\mathcal S}\rangle)&\nonumber \\
 &= \sum_{m=1}^{k}\max_{|\psi^{r_m} \rangle  \in  \mathcal H^{(r_m)}}\left(\mathcal M^{(pb)}(|\psi^{r_m}\rangle)\right)&\nonumber \\
&= \sum_{m=1}^k {r_m \choose 2}+d,&\\ \nonumber
\end{eqnarray}
where $\mathcal H^{(r_m)}$ is the Hilbert space for subsystems $r_m$.
The first equality follows from Proposition \ref{prop sum} and the second from Proposition \ref{teor M}. The additional term $d$ comes from the fact that whenever $r_m=2$, the maximum of $\mathcal   M^{(pb)}(|\psi^{r_m}\rangle)$ is equal to $2$, and not ${2 \choose 2}=1$. $\Box$

 If $\mathcal M^{(pb)}(|\psi\rangle) > \sum_{m=1}^k {r_m \choose 2}+d$ then we know, thanks to the above criterion, that  the state cannot be a $k$-product state with respect to  the division  $r_1,\dots,r_k$.
The following Lemma  can be used to show  the maximum of the bound (\ref{PROP_KSEP}) over all possible splittings of $n$ parties into $k$ subsystems, i.e. all $k$-element decompositions of number $n$ into natural numbers.
\begin{lem}
For any $n\geq 3$ and for any $k$ from $2$ to $n$, and for any sequence $\{r_m\}_{m=1}^k$ such that  and $\sum_{m=1}^k r_m =n$, the following conditions hold:
\begin{itemize}
	\item if all  $r_m\geq 1$, one has:
	$
	\sum_{m=1}^k {r_m \choose 2}\leq {n-k+1 \choose 2} ,
	$
	\item if all  $r_m\geq 0$, one has:
		$
	\sum_{m=1}^k {r_m \choose 2}\leq {n \choose 2}.
	$
	
\end{itemize}
\label{lem1}
\end{lem}
{\em Proof.} Note that the sum on the left-hand side of the above inequalities  can be expressed as:
\begin{equation}
\sum_{m=1}^k {r_m \choose 2}=\sum_{m=1}^k \frac{r_m(r_m-1)}{2}=\frac{1}{2}\left(\sum_{m=1}^k r_m^2-n\right)
\end{equation}
Let us consider the  case with all $r_m\geq 1$. The sum $\sum_{m=1}^k r_m^2$ is maximized, when all but one of the terms in the sequence $\{r_m\}_{m=1}^k$ are  equal to 1, and the remaing one  equals $n-k+1$.
This can be easily  proved using Lagrange method of finding conditional extrema. 

In the case of $0\leq r_m \leq n$ we can prove that with given boundary conditions, the sum is maximized when all but one terms are $0$, and the one equals $n$. $\Box$

Results of the Lemma and Proposition \ref{teor k-sep} imply the following corollaries:
\begin{propo}
For $n\geq 3$, and  for any pure $n$-qubit state $|\psi\rangle$, the following implication holds. If
$\mathcal M^{(pb)}(|\psi\rangle)>s_k$ then $|\psi\rangle \textrm{ is not k-product},$
where  $s_k$ for $k=n-1$ is $2$, for $k=n-2$ is $4$, and
\begin{equation}
s_{k}={n-k+1 \choose 2}, \textrm{ for } k=2,...,n-3.
\end{equation}
\label{s seq}
\end{propo}

\begin{propo}
For $n\geq 5$, and for any n-qubit pure state $|\psi\rangle$,  if
$$\mathcal M^{(pb)}(|\psi\rangle)>{n-1 \choose 2}$$
then $|\psi\rangle$ is genuinely n-partite entangled.
For $n=3$ and $n=4$ we have respectively:
$$\mathcal M^{(pb)}(|\psi\rangle)>2 \Longrightarrow |\psi\rangle \textrm{ is genuinely 3-partite entangled.}$$
$$\mathcal M^{(pb)}(|\psi\rangle)>4 \Longrightarrow |\psi\rangle \textrm{ is genuinely 4-partite entangled.}$$
\label{teor main}
\end{propo}

Our last condition for multipartite entanglement is given by the following proposition.
\begin{propo}
For any n-qubit pure state $|\psi\rangle$, with $n\geq 5$, and for any $m\leq \left\lfloor \frac{n}{2}\right\rfloor-1$ the following holds:
   \begin{eqnarray*}
   &&\mathcal M^{(pb)}(|\psi\rangle)>{m \choose 2}+{n-m \choose 2}+\delta_{m,2}  \nonumber\\
   &&\Longrightarrow |\psi\rangle \textrm{ is genuinely $m$--partite entangled.}
   \end{eqnarray*}
\label{teor main1}
\end{propo}

\emph{Proof}.
For any $n$-partite state consider all possible bipartite divisions $(r_1+r_2)$, where $r_1\leq r_2$ and $r_1+r_2=n$. Note that the $\mathcal M^{(pb)}$ values corresponding to different divisions are in the following order:
\begin{eqnarray}
\mathcal M^{(pb)}\left(\ket{\psi^1} \ket{\psi^{n-1}}\right)>\mathcal M^{(pb)}\left(\ket{\psi^2} \ket{\psi^{n-2}}\right)>... \nonumber \\>\mathcal M^{(pb)}\left(\ket{\psi^{\lfloor n/2 \rfloor}} \ket{\psi^{\left\lceil n/2\right\rceil}}\right).
\end{eqnarray}
Indeed, for any $r_1\in \left[1,\left\lfloor \frac{n}{2} \right\rfloor-1\right]$ and $n \ge 5$ (such that division $(2+2)$ is excluded) the difference of adjacent divisions is strictly positive:
\begin{eqnarray}
&&\mathcal M^{(pb)}\left(\ket{\psi^{r_1}} \ket{\psi^{n-r_1}}\right)-\mathcal M^{(pb)}\left(\ket{\psi^{r_1+1}} \ket{\psi^{n-r_1-1}}\right) \nonumber \\
&=&{r_1 \choose 2}+\delta_{r_1,2}+{n-r_1 \choose 2}+\delta_{n-r_1,2} \nonumber \\
&-&{r_1+1 \choose 2}-\delta_{r_1+1,2}-{n-r_1-1 \choose 2}-\delta_{n-r_1-1,2} \nonumber \\
&\ge & n-2r_1-2 > 0
\end{eqnarray}
If $\mathcal M^{(pb)}(|\psi\rangle)>\mathcal M^{(pb)}\left(\ket{\psi^m} \ket{\psi^{n-m}}\right)$ then, due to the ordering of divisions, the state $\psi$ contains entanglement between at least $m+1$ parties.$\Box$


As an example of an application of the conditions
let us consider a family of $n$-partite Dicke states with $e$ excitations:
\begin{equation}
|D_n^{e}\rangle=\frac{1}{\sqrt{{n \choose e}}}\sum_{\pi}|\pi(\underbrace{1,...,1}_{e},\underbrace{0,...,0}_{n-e})\rangle,
\label{DickeNe}
\end{equation}
where summation is performed over all combinations. It can be  shown that for odd $n$
\begin{equation}
\mathcal M^{(pb)}(|D_n^e\rangle)=\frac{4e^2(n-e)^2}{n(n-1)}.
\label{Mne}
\end{equation}
This expression is maximized for states of the type $D_n^{(n-1)/2}$ for which
\begin{equation}
\mathcal M^{(pb)}\left(|D_n^{(n-1)/2}\rangle\right)=\frac{(n+1)^2 (n-1)}{4n}.
\label{Mnmax}
\end{equation}
Using our conditions for entanglement we can prove the following properties of the family.
States $D_3^1$ and $D_5^2$ are genuinely multipartite entangled.  This follows from Proposition \ref{teor main}.
States $D_n^{(n-1)/2}$ for $n$ is odd, are genuinely $\left(\frac{n+3}{2}\right)$--partite entangled. This is a consequence of Proposition \ref{teor main1}.


The fact that the  multipartite entanglement of Dicke states can be detected with lower order correlations is itself not surprising, since it is known that these states are uniquely determined by their reduced density matrices \cite{PR_PRA09,PR09}.
Namely, a Dicke state with $e$ excitations is the only state compatible with its $2e$-partite reduced density operators \cite{PR_PRA09,PR09}.  With an increasing  $e$ one has a  growing dependence on multipartite correlations.
It is thus intriguing that the conditions we derive, although solely based on bipartite correlations, detect better  entanglement of Dicke states with higher excitation number.
One reason for this could be that entanglement is just one of many defining features of a quantum state and for this reason knowledge of the full state is not necessary to detect entanglement. 

A different set of conditions based on bipartite correlations,  proposed in ref. \cite{GEZA},  uses spin operators.
They are linear in correlations whereas ours are quadratic.
Quadratic conditions have an advantage: once the bound for entanglement is achieved we can stop measuring  as any future measurement can only make increase the sum of squared correlations.
For the linear conditions all correlations have to be measured as  correlations measured in the next experimental setting could be negative.

Note that  Dicke states are genuinely $n$-party entangled. The fact that our criterion does not detect this seems to be a drawback.
A question emerges: could bipartite correlations be used to detect this entanglement?
Note that for a biproduct state there exist two subsystems, say $k$-th and $l$-th with correlations satisfying $T_{ij}=b_j^{(k)}b_j^{(l)}$ as every correlation is just a product of local averages.
Therefore, if Dicke state were biproduct then some of its bipartite correlations would have to satisfy this relation. 

We would like to stress that the conditions derived here do not make any additional assumptions about studied states, except  purity, and the Dicke states should be treated just as a simple example of their application.


TP acknowledges discussions with Minh Cong Tran.
The work is supported by Polish Ministry of Science and Higher
Education Grant no. IdP2011 000361.
MM is supported by  International PhD
Project ``Physics of future quantum-based information technologies''
grant MPD/2009-3/4 of Foundation for Polish Science (FNP).  WL and M\.Z
are supported by  FNP
TEAM project co-financed by  EU Regional Development Fund.
TP is supported by  National Research Foundation, Ministry of Education of Singapore, and NTU start-up grant.

\end{document}